\documentclass[11pt]{article}
\usepackage{titlesec}
\usepackage[toc,page]{appendix}
\usepackage{amsmath}

\usepackage{amsthm}
\usepackage{algorithm}
\usepackage{algorithmicx}
\usepackage[noend]{algpseudocode}
\usepackage[bottom]{footmisc}
\usepackage{tikz}
\usetikzlibrary{decorations.text}
\usepackage{multicol}
\usepackage{mathtools}
\usepackage{adjustbox}
\usepackage{diagbox}
\usepackage{array}
\newcolumntype{P}[1]{>{\centering\arraybackslash}p{#1}}
\newcolumntype{M}[1]{>{\centering\arraybackslash}m{#1}}

\makeatletter
\def\BState{\State\hskip-\ALG@thistlm}
\makeatother
\usetikzlibrary{calc, chains,
                fit,
                positioning,
                shapes}

\definecolor{myblue}{RGB}{80,80,160}
\definecolor{mygreen}{RGB}{80,160,80}
\definecolor{myempty}{RGB}{255,255,255}

\usetikzlibrary{shapes,arrows,positioning,decorations.markings}

\usepackage{graphicx}
\usepackage{hyperref}
\hypersetup{
    bookmarks=true,         
    unicode=false,          
    pdftoolbar=true,        
    pdfmenubar=true,        
    pdffitwindow=false,     
    pdftitle={Certificate},    
    pdfauthor={Gal Oren, Leonid Barenboim},     
    pdfsubject={Distributed Backup Placement in One-Round and its Applications to Maximum Matching Approximation},   
    pdfcreator={Gal Oren, Leonid Barenboim},   
    pdfproducer={},  
    pdfkeywords={},
    pdfnewwindow=true,      
    colorlinks=false,       
    linkcolor=red,          
    citecolor=green,        
    filecolor=magenta,      
    urlcolor=cyan           
}
\usepackage[utf8]{inputenc}
\usepackage[english]{babel}
\newtheorem{theorem}{Theorem}[section]

\newtheorem{lemma}[theorem]{Lemma}

\def\denseformat{
	\setlength{\textheight}{9in}
	\setlength{\textwidth}{6.9in}
	\setlength{\evensidemargin}{-0.2in}
	\setlength{\oddsidemargin}{-0.2in}
	\setlength{\headsep}{10pt}
	\setlength{\topmargin}{-0.3in}
	\setlength{\columnsep}{0.375in}
	\setlength{\itemsep}{0pt}
}

\denseformat

\titlespacing*{\section}
{0pt}{2ex}{0.5ex}
\titlespacing*{\subsection}
{0pt}{0ex}{0.5ex}

\begin{document}
\title{Distributed Backup Placement in One Round and its Applications \\ to Maximum Matching Approximation and Self-Stabilization}

\author{Leonid Barenboim \thanks{Open University of Israel. E-mail: 
        {\tt\small leonidb@openu.ac.il}} \ , \ Gal Oren \thanks{Ben-Gurion University of the Negev, Nuclear Research Center - Negev. E-mail:
        {\tt\small orenw@post.bgu.ac.il} \newline This work was supported by the Lynn and William Frankel Center for Computer Science, the Open University of Israel's Research Fund, and ISF grant 724/15.}
}





\date{}

\def\thepage{}
\maketitle 
\begin{abstract}

In the distributed backup-placement problem \cite{halldorsson2015bp} each node of a network has to select one neighbor, such that the maximum number of nodes that make the same selection is minimized. This is a natural relaxation of the perfect matching problem, in which each node is selected just by one neighbor. Previous (approximate) solutions for backup placement are non-trivial, even for simple graph topologies, such as dense graphs \cite{barenboim2019fast}. In this paper we devise an algorithm for dense graph topologies, including unit disk graphs, unit ball graphs, line graphs, graphs with bounded diversity, and many more. Our algorithm requires just one round, and is as simple as the following operation. Consider a circular list of neighborhood IDs, sorted in an ascending order, and select the ID that is next to the selecting vertex ID. Surprisingly, such a simple one-round strategy turns out to be very efficient for backup placement computation in dense networks. Not only that it improves the number of rounds of the solution, but also the approximation ratio is improved by a multiplicative factor of at least $2$.

Our new algorithm has several interesting implications. In particular, it gives rise to a $(2 + \epsilon)$-approximation to maximum matching within $O(\log^* n)$ rounds in dense networks. The resulting algorithm is very simple as well, in sharp contrast to previous algorithms that compute such a solution within this running time. Moreover, these algorithms are applicable to a narrower graph family than our algorithm. For the same graph family, the best previously-known result has $O(\log {\Delta} + \log^* n)$ running time \cite{barenboim2018distributed}. Another interesting implication is the possibility to execute our backup placement algorithm as-is in the self-stabilizing setting. This makes it possible to simplify and improve other algorithms for the self-stabilizing setting, by employing  helpful properties of backup placement. 
\end{abstract}

\section{Introduction}
\label{intro1}

{\bf Background:} The perfect matching problem in a graph $G = (V,E)$ aims to find a subset of edges $E' \subseteq E$, such that the edges of $E'$ are non-adjacent, and each vertex of $V$ belongs to exactly one edge of $E'$. Unfortunately, not every graph admits a perfect matching. In this paper we consider a natural relaxation of perfect matchings. Specifically, each vertex must select a neighbor, such that the maximum number of vertices that select the same vertex in the graph is minimized. In other words, the goal is finding a subset of edges $E' \subseteq E$, such that each vertex of $v$ belongs to at least on edge of $E'$, and the maximum degree in $G' = (V,E')$ is minimized. In addition, the edges of $E'$ are oriented, and each vertex in $G'$ has out degree of one. This corresponds to the requirement that each vertex selects just one neighbor, but can be selected by several neighbors. Note that if $E'$ is a maximum matching, the maximum degree in $G'$ is 1. When finding a maximum matching is not possible, the goal is minimizing the maximum degree of $G'$, while making sure that all vertices belong to edges of $E'$.

In the distributed setting, this problem is called {\em Backup Placement}. It was introduced by  Halldorsson, Kohler, Patt-Shamir, and Rawitz \cite{halldorsson2015bp}. It is very well motivated by computer networks whose nodes may have memory faults, and wish to store backup copies of their data at neighboring nodes \cite{oren2018distributed}. But neighboring nodes may incur faults as well, and so the number of nodes that select the same backup-node should be minimized. This way, if a backup node incurs faults, the number of nodes in the network that lose data is minimized.

The precise definition of the distributed variant of the problem is as follows. The computer network is represented by an unweighted unoriented graph $G = (V,E)$, where $V$ is the set of nodes, and $E$ is the set of communication links. Communication proceeds in synchronous discrete rounds. In each round vertices receive and send message, and perform local computations. A message sent over an edge in a certain round arrives to the endpoint of that edge by the end of the round. The algorithm terminates once every vertex outputs its neighbor selection for the backup placement. The running time is the number of rounds from the beginning until all vertices make their decisions. We consider two variants of networks: faultless networks and faulty networks. For the latter, the goal is obtaining a self-stabilizing algorithm. In Sections \ref{intro2} - \ref{intro3} we consider faultless networks. In Section \ref{selfstab} we consider faulty networks, and elaborate there on the additional properties of the problem and the required solution in this case.

The backup placement problem turned out to be quite challenging in general graphs. The best-currently known solution is a randomized algorithm with running time $O(\frac{\log^6 n}{\log^4 \log n})$ that obtains an approximation factor of $O(\frac{\log n}{\log \log n})$. This solution due to Halldorsen et al. \cite{halldorsson2015bp} is non-trivial at all, and involves distributed computations of a certain variant of matching, called an $f$-matching, in bipartite graphs. On the other hand, in certain network topologies, simpler and much more efficient solutions are known. In particular, this is the case in wireless networks, certain social networks, claw-free graphs, line graphs, and more generally, any graph with {\em neighborhood-independence bounded by a constant $c$}.  Neighborhood independence $I(G)$ of a graph $G=(V,E)$ is the maximum size of an independent set contained in a neighborhood $\Gamma(v), v \in V$. For graphs with $I(G) \leq c = O(1)$, a constant-time deterministic distributed algorithm with approximation ratio $2c + 1 = O(1)$ was devised by Barenboim and Oren \cite{barenboim2019fast}. Although not so complicated, this algorithm still consists of several stages, including a computation of a tree cover, and then handling differently the different parts of the trees. (Such as leafs and non-leafs.)

{\bf Our Results:} In the current paper we significantly simplify the backup placement algorithm for graphs with neighborhood independence bounded by a constant $c$. Specifically, the algorithm becomes uniform, and consists just of a single instruction that should be executed by all nodes in parallel within a single round. Consequently, the running time becomes just one round, which improves the  number of rounds required in the previous solution for such graphs by a constant. More importantly, this improves the approximation ratio as well, which becomes $c$, rather than $2c + 1$ of \cite{barenboim2019fast}. Furthermore, this instruction is solely a function of the IDs of a vertex and its neighbors. As IDs are stored in areas that are considered to be failure-free (in contrast to variables that are stored in Random Access Memory that is failure-pron), algorithms that perform computations only as a function of IDs within a single round can be translated into self-stabilizing ones in a straightforward way. The structure of our algorithm makes it especially suitable for implementation in real-life networks with limited resources, such as sensor networks, heterogeneous networks, and Internet of Things.

We employ our backup-placement algorithm in order to compute {\em maximum matching approximation} of an input graph $G$ with neighborhood independence bounded by $c$. For $c = O(1)$, we obtain a $(2 + \epsilon)$-approximation to maximum matching within $O(\log^* n)$ rounds. The best previously-known $O(1)$-approximation for such graph has running time $O(\log{\Delta} + \log^*n)$ \cite{barenboim2018distributed}, where $\Delta$ is the maximum degree of the graph. Another $O(1)$-approximate matching result, for a narrower family of graphs with bounded growth, has running time $O(\log^* n)$ \cite{schneider2008log}. However, this result of Schneider and Wattenhofer \cite{schneider2008log} is based on network decompositions, whose computation in such graphs involves very sophisticated arguments. Our algorithm, on the other hand, applies to a wider family of graphs, and is very simple. Specifically, it performs a constant number of iterations, each of which consists of computing a backup-placement $G' = (V,E')$, computing a maximal matching of it, and removing the matched edges and edges adjacent on them from $G$. Since the maximum degree of $G' = (V,E')$ is $c + 1 = O(1)$, computing a maximal matching in it requires just $O(\log^* n)$ rounds, using \cite{panconesi2001some}.

{\bf Graphs with bounded independence:} The family of graphs with neighborhood independence bounded by a constant is a very wide family that captures various network types. This includes unit disk graphs, unit balls graphs, line graphs, line graphs of hypergraphs, claw-free graphs, graphs of bounded diversity, and many more. Consequently, this graph family and its subtypes have been very widely studied, especially in the distributed computing setting \cite{gfeller2007randomized, schneider2009coloring, barenboim2011deterministic, barenboim2017deterministic, barenboim2018distributed, assadi2019algorithms, kuhn2019faster}. For example, unit disc graphs can model certain types of wireless networks. In such networks all nodes have the same transmission range that is the radius of a disc. If nodes are positioned in the plane, the neighbors of any node can be covered by at most 6 discs of radius $1/2$. Each such disc forms a clique, since all nodes inside it can transmit one to another. Thus a disc of radius $1/2$ cannot contain two or more independent nodes. Hence, the neighborhood independence of unit disc graphs is at most $6$. Another notable example is the family of line graphs. In these graphs each vertex belongs to at most $2$ cliques, and thus the neighborhood independence is bounded by $2$. 

A notable example of the benefit of analyzing such graphs is the very recent breakthrough of Kuhn \cite{kuhn2019faster}. Kuhn obtained a $(2\Delta-1)$-edge-coloring of {\em general graphs} by analyzing graphs with constant neighborhood independence. The resulting algorithm provides a vertex coloring of such graph with time below the square-root-in-$\Delta$ barrier. Since this provides, in particular, a vertex coloring of line graphs, it results in an edge coloring of general graphs. This result, as well as other results for this topology, illustrate how beneficial can be the analysis of graphs with bounded neighborhood independence.

\section{Distributed Backup Placement Algorithm}\label{intro2}

We begin with devising a procedure for computing $O(1)$-backup placement in graphs with bounded \textit{neighborhood independence} \textit{c}. 
We assume that each vertex knows only about its neighbors, and each vertex has a unique ID. The procedure receives a graph $G = (V,E)$ as input, and proceeds as follows. 
We define an operation {\em next-modulo} that receives a vertex $v$ and a set of its neighbors $\Gamma(v)$ in the graph $G$. The operation \textit{next-modulo}($v$, $\Gamma(v)$), selects a vertex in $\Gamma(v)$ with a higher ID than the ID of $v$, and whose ID is the closets to that of $v$. If no such neighbor is found, then the operation returns the neighbor with the smallest ID. 
Formally, each vertex $v \in V$ selects a neighbor $w$ of $v$ in $G$ with the property that $ID(w) > ID(v)$, and there is no other neighbor $z$ of $v$ such that $ID(w) > ID(z) > ID(v)$. If there is no such neighbor, then $v$ selects the minimum ID vertex in $\Gamma(v)$. All these selections are performed in parallel within a single round. This completes the description of the algorithm. Its pseudocode is provided in Algorithm \ref{algo1}. Its action is illustrated in Figure \ref{fig1}. The next theorem summarizes its correctness.

\begin{algorithm}[H]
\caption{Backup Placement in Graphs}
\label{algo1}
\begin{algorithmic}[1]
\Procedure{Graph-BP(Graph $G = (V,E)$)}{}
\State {\bf foreach} node $v \in G$ in parallel do:
\State \hspace{0.5cm} v.BP = \textit{next-modulo}($v$, $\Gamma(v)$) \State /* find the vertex next to $v$ in $\Gamma(v)$, according to a circular list of sorted IDs of $\Gamma(v) \cup v$ */

\EndProcedure
\end{algorithmic}
\end{algorithm}

\begin{figure}[H]
\centering
\begin{tikzpicture}

  \tikzstyle{vertex}=[circle,fill=black!25,minimum size=12pt,inner sep=2pt]
  
  \node[vertex] (G_0) at (0,0) {25};
  \node[vertex] (G_1) at (0.5,2) {9};
  \node[vertex] (G_2) at (2,0.5) {30};
  \node[vertex] (G_3) at (2,-0.5) {4};
  \node[vertex] (G_4) at (0.5,-2) {40};
  \node[vertex] (G_5) at (-0.5,-2) {7};
  \node[vertex] (G_6) at (-2,-0.5) {50};
  \node[vertex] (G_7) at (-2,0.5) {6};
  \node[vertex] (G_8) at (-0.5,2) {20};

  \draw [line width=0.2mm, black] (G_0) -- (G_1); 
  \draw [line width=0.2mm, black] (G_1) -- (G_8); 
  \draw [line width=0.2mm, black] (G_8) -- (G_0); 

  \draw [line width=0.2mm, black] (G_0) -- (G_3); 
  \draw [line width=0.2mm, black] (G_3) -- (G_2); 
  \draw [line width=0.2mm, black] (G_2) -- (G_0); 
  
  \draw [line width=0.2mm, black] (G_0) -- (G_5); 
  \draw [line width=0.2mm, black] (G_5) -- (G_4); 
  \draw [line width=0.2mm, black] (G_4) -- (G_0); 
  
  \draw [line width=0.2mm, black] (G_0) -- (G_6); 
  \draw [line width=0.2mm, black] (G_6) -- (G_7); 
  \draw [line width=0.2mm, black] (G_7) -- (G_0);

  \draw [->, line width=0.2mm, blue] [dashed] (G_1) to[out=135,in=45] (G_8);
  \draw [->, line width=0.2mm, blue] [dashed] (G_8) to[out=225,in=135] (G_0);

  \draw [->, line width=0.2mm, blue] [dashed] (G_3) to[out=225,in=315] (G_0);
  \draw [->, line width=0.2mm, blue] [dashed] (G_2) to[out=315,in=45] (G_3);
  \draw [->, line width=0.2mm, blue] [dashed] (G_0) to[out=45,in=135] (G_2);
 
  \draw [->, line width=0.2mm, blue] [dashed] (G_5) to[out=135,in=225] (G_0);
  \draw [->, line width=0.2mm, blue] [dashed] (G_4) to[out=225,in=315] (G_5); 

  \draw [->, line width=0.2mm, blue] [dashed] (G_7) to[out=45,in=135] (G_0);
  \draw [->, line width=0.2mm, blue] [dashed] (G_6) to[out=135,in=225] (G_7);

\end{tikzpicture}
\caption{A backup placement (\textit{blue}) in a graph with bounded neighborhood independence \textit{c = 4}.} 
\label{fig1}
\end{figure}
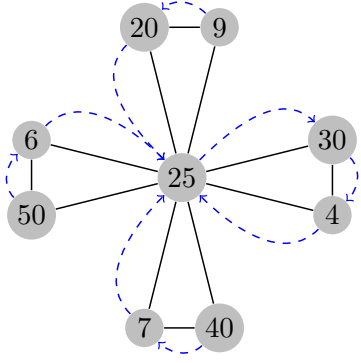

\begin{theorem}
Let $G$ be a graph with neighborhood independence of $c$, i.e., for each vertex in the graph with \textit{c+1} neighbors, at least two neighbors are connected by an edge. Then each vertex is selected by at most $c$ neighbors.
\end{theorem}
\begin{proof}
Assume for contradiction that $c+1$ vertices chose the same vertex $u$ as a backup placement. Since the neighborhood independence is $c$, at least two vertices, $v_1, v_2$, of the $c+1$ vertices, are connected by an edge. Therefore, a triangle $\{u, v_1, v_2\}$ is formed in the subgraph induced by the edges selected by the algorithm. However, we next show that $u$ could not have been selected by both its neighbors $v_1,v_2$ in the triangle:
According to Algorithm \ref{algo1}, there are only three possibilities regarding ID order: \\
(1) $ID(u) < ID(v_1)$ and $ID(u) < ID(v_2)$, \\
(2) $ID(v_1) < ID(u)$ and $ID(v_2) < ID(u)$, \\
(3) $ID(v_1) < ID(u) < ID(v_2)$ or $ID(v_2) < ID(u) < ID(v_1)$.

However, in each of those possibilities, as presented in the three figures below, it is impossible for two vertices to choose the same vertex for backup placement. This is because there is always a vertex between $v_1$ and $u$ or between $v_2$ and $u$, that one of $u_1,u_2$ must select. Indeed, either $u_1$ is closer to $u_2$ than to $v$, with respect to \textit{next-modulo} operation, or $u_2$ is closer to $u_1$. 
More precisely:
\begin{itemize}
  \item Case (1),(2): if $ID(v_1) < ID(v_2)$ then $v_1$ selects $v_2$ or another neighbor with ID smaller than that of $v_2$, but not $u$. Otherwise $v_2$ selects $v_1$, or a neighbor with a smaller ID than that of $v_1$, but not $u$. In any case, both cannot select $u$ simultaneously.
  \item Case (3): The vertex with the highest ID among the three is either $v_1$ or $v_2$. It must select a neighbor with an even greater ID or the minimum ID. But $u$ has an ID smaller than that of $v_1$ or $v_2$, and is not the minimal among ${u,v_1,v_2}$. Thus, either $v_1$ or $v_2$ does not select $v$.
\end{itemize}

In any case, we have a contradiction to the assumption that both $v_1$,$v_2$ select $u$.

\begin{figure}[H]
    \centering
    
        \begin{minipage}[b]{0.3\textwidth}
        \centering
        
\begin{tikzpicture}
  \tikzstyle{vertex}=[circle,fill=black!25,minimum size=12pt,inner sep=2pt]
  
  \node[vertex] (G_0) at (0,0) {1};
  \node[vertex] (G_1) at (-1.5,-2) {2};
  \node[vertex] (G_2) at (1.5,-2) {3};
  
  \draw [-, line width=0.3mm, black] (G_0) -- (G_1); 
  \draw [-, line width=0.3mm, black] (G_0) -- (G_2);
  \draw [-, line width=0.3mm, black] (G_1) -- (G_2);

  \def\myshift#1{\raisebox{-2.5ex}}
\draw [->,line width=0.3mm, blue, dashed] (G_0) to [bend right=45]  (G_1);

  \def\myshift#1{\raisebox{-2.5ex}}
\draw [->,line width=0.3mm, blue, dashed] (G_1) to [bend right=45]  (G_2);

  \def\myshift#1{\raisebox{-2.5ex}}
\draw [->,line width=0.3mm, blue, dashed,postaction={decorate,decoration={text along path,text align=center,text={|\myshift|next-modulo}}}] (G_2) to [bend right=45]  (G_0);
  
\end{tikzpicture}
        \caption*{$2=ID(v_1) > ID(u)=1$\\$3=ID(v_2) > ID(u)=1$}
        \label{fig1}
    \end{minipage}
    \hspace{-5mm}
        \begin{minipage}[b]{0.3\textwidth}
        \centering
\begin{tikzpicture}
  \tikzstyle{vertex}=[circle,fill=black!25,minimum size=12pt,inner sep=2pt]
  
  \node[vertex] (G_0) at (0,0) {3};
  \node[vertex] (G_1) at (-1.5,-2) {1};
  \node[vertex] (G_2) at (1.5,-2) {2};
  
  \draw [-, line width=0.3mm, black] (G_0) -- (G_1); 
  \draw [-, line width=0.3mm, black] (G_0) -- (G_2);
  \draw [-, line width=0.3mm, black] (G_1) -- (G_2);

  \def\myshift#1{\raisebox{-2.5ex}}
\draw [->,line width=0.3mm, blue, dashed,postaction={decorate,decoration={text along path,text align=center,text={|\myshift|
next-modulo}}}] (G_0) to [bend right=45]  (G_1);

  \def\myshift#1{\raisebox{-2.5ex}}
\draw [->,line width=0.3mm, blue, dashed] (G_1) to [bend right=45]  (G_2);

  \def\myshift#1{\raisebox{-2.5ex}}
\draw [->,line width=0.3mm, blue, dashed] (G_2) to [bend right=45]  (G_0);

  
\end{tikzpicture}
        \caption*{$1=ID(v_1) < ID(u)=3$\\$2=ID(v_2) < ID(u)=3$}
        \label{fig2}
    \end{minipage}
\hspace{-5mm}
        \begin{minipage}[b]{0.3\textwidth}
        \centering
\begin{tikzpicture}
  \tikzstyle{vertex}=[circle,fill=black!25,minimum size=12pt,inner sep=2pt]
  
  \node[vertex] (G_0) at (0,0) {2};
  \node[vertex] (G_1) at (-1.5,-2) {1};
  \node[vertex] (G_2) at (1.5,-2) {3};
  
  \draw [-, line width=0.3mm, black] (G_0) -- (G_1); 
  \draw [-, line width=0.3mm, black] (G_0) -- (G_2);
  \draw [-, line width=0.3mm, black] (G_1) -- (G_2);

  \def\myshift#1{\raisebox{-2.5ex}}
\draw [->,line width=0.3mm, blue, dashed,postaction={decorate,decoration={text along path,text align=center,text={|\myshift|
}}}] (G_1) to [bend left=45]  (G_0);

  \def\myshift#1{\raisebox{-2.5ex}}
\draw [->,line width=0.3mm, blue, dashed] (G_2) to [bend left=45]  (G_1);

  \def\myshift#1{\raisebox{-2.5ex}}
\draw [->,line width=0.3mm, blue, dashed] (G_0) to [bend left=45]  (G_2);

  \def\myshift#1{\raisebox{-2.5ex}}
\draw [-,line width=0, transparent!0, dashed,postaction={decorate,decoration={text along path,text align=center,text={|\myshift|next-modulo}}}] (G_1) to [bend right=45]  (G_2);
  
\end{tikzpicture}
        \caption*{$ID(v_1) < ID(u) < ID(v_2)$\\ $1$ \ \ $ < $ \ \ $ 2 $ \ \ $ < $ \ \ $3$}
        \label{fig3}
    \end{minipage}
    \label{fig_proof1}
\end{figure}
\end{proof}

\section{Maximum Matching Approximation based on Backup Placement}\label{intro3}
A set of edges $M \in E$ is called a \textit{Matching} if and only if every vertex $v \in V$ belongs to at most one edge in $M$. 
A \textit{Maximal Matching} (shortly, \textit{MM}) is a matching that is maximal with respect to addition of edges, i.e., there is no edge $e \in E$ such that $M \cup \{e\}$ is a valid matching. A \textit{Maximum Matching} (shortly, \textit{MCM}) is a matching of maximum size among all matchings of $E$.

As shown in the previous section, given a graph $G=(V,E)$ with bounded neighborhood independence $c$, we can compute a $O(1)$-backup placement in a single round. This results in a subgraph, $G'=(V,E')$, where $E'$ is the set of all the selected edges of the backup placement algorithm. Each vertex in the subgraph $G'$ has selected one neighbor, and is selected by at most $c$ neighbors. All edges adjacent on a vertex in $G'$ are either selecting edges or selected edges. Thus, the number of such edges is at most $c + 1$. Consequently, the maximum degree $\Delta(G')$ is at most $c + 1$.

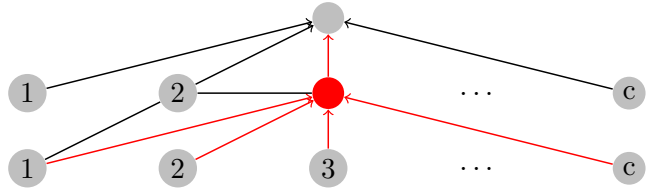
\begin{figure}[H]
\centering
\begin{tikzpicture}
  \tikzstyle{vertex}=[circle,fill=black!25,minimum size=12pt,inner sep=2pt]
  \tikzstyle{vertex_red}=[circle,fill=red,minimum size=12pt,inner sep=2pt]
  

  \node[vertex] (G_0_1) at (0,1) {};

  \node[vertex] (G_2_1) at (-4,0) {1};
  \node[vertex] (G_6_1) at (-2,0) {2};
  \node         (G_4_1) at (2,0) {\ldots};
  \node[vertex] (G_5_1) at (4,0) {c};

  \node[vertex_red] (G_1) at (0,0) {};
  
  \node[vertex] (G_2) at (-4,-1) {1};
  \node[vertex] (G_6) at (-2,-1) {2};
  \node[vertex] (G_3) at (0,-1) {3};
  \node         (G_4) at (2,-1) {\ldots};
  \node[vertex] (G_5) at (4,-1) {c};
  
  \draw [->, line width=0.2mm, red] (G_2) -- (G_1);
  \draw [->, line width=0.2mm, red] (G_6) -- (G_1);
  \draw [->, line width=0.2mm, red] (G_3) -- (G_1);
  \draw [->, line width=0.2mm, red] (G_5) -- (G_1);

  \draw [->, line width=0.2mm, black] (G_2_1) -- (G_0_1);
  \draw [->, line width=0.2mm, black] (G_6_1) -- (G_0_1);
  \draw [->, line width=0.2mm, red] (G_1) -- (G_0_1);
  \draw [->, line width=0.2mm, black] (G_5_1) -- (G_0_1);
  
  \draw [-, line width=0.2mm, black] (G_6_1) -- (G_2);
  \draw [-, line width=0.2mm, black] (G_6_1) -- (G_1);
  
  \end{tikzpicture}
\caption{The maximum degree $\Delta(G')$ of the subgraph $G'=(V,E')$ is at most $c + 1$. In red: a vertex with degree of $c + 1$, as it was selected by its neighbors with neighborhood independence of $c$, and selected one additional neighbor.} \label{fig3}
\end{figure}

We devise a Maximum Matching approximation based on this backup placement subgraph in $O(\log^* n)$ rounds. To this end, we compute a backup placement of an input graph $G$, obtain the graph $G' = (V,E')$, where $E'$ is the set of selected edges, and execute a maximal matching algorithm of Panconesi and Rizzi \cite{panconesi2001some} on $G'$. The latter algorithm computes a maximal matching of an input graph with degree $\Delta$ within $O(\Delta + \log^*n)$ rounds.

\begin{lemma}
\label{lemmaa}
Given a graph $G=(V,E)$ with bounded neighborhood independence $c$, we achieve $(c + 1)$-approximation of the Maximum Matching problem.
\end{lemma}
\begin{proof}
We begin with executing the backup placement algorithm, which computes the sub-graph $G'=(V,E')$. Next, we compute an MM of $G'$. Since $G'$ has bounded neighborhood independence $c$ and $\Delta(G') = c + 1$, we can show that a maximal matching of $G'$ has size at least $1/(c + 1)$ of $MCM(G)$.
This is because every vertex in $V$ is either in $MM(G')$ or adjacent (in $G'$) to a vertex in $MM(G')$. (Otherwise, it is adjacent to a free vertex, and an edge can be added to the maximal matching of $G'$. Contradiction.)  Thus, the set of vertices that belong to edges of $MM(G')$ together with the vertices adjacent on them in $G'$, are exactly the set $V$. On the other hand, since each vertex of $MM(G')$ is adjacent in $G'$ to at most $c + 1$ vertices, the size of $V$ is at most $c + 1$ times the number of vertices of $MM(G')$. Since the size of the maximum matching is at most $|V|/2$ and the size of $MM(G')$ is at least $|V|/2(c + 1)$, we obtained a $(c + 1)$-approximation to maximum matching


\end{proof}

\begin{lemma}
Given a graph $G=(V,E)$ with bounded neighborhood independence $c$, the running time of the Maximum Matching approximation is $O(\log^* n)$.
\end{lemma}
\begin{proof}
Using $O(\Delta + \log^* n)$-time Maximal Matching algorithm by Panconesi and Rizzi \cite{panconesi2001some}, and due to the fact that $\Delta(G') = c + 1 = O(1)$, the achieved time complexity is $O(\log^* n)$.
\end{proof}

In order to reach an even better Maximum Matching approximation than the $(c + 1)$-approximation, we apply several times the maximal matching algorithm by Panconesi and Rizzi \cite{panconesi2001some} on $G'$ which was obtained by the $O(1)$-backup placement in graphs with bounded \textit{neighborhood independence} \textit{c}. In order to preserve a proper matching of $G$ in each step, after each computation of a maximal matching, we remove the resulting edges endpoints from $G'$, as well as all edges adjacent on these endpoints. Then we invoke again a maximal matching computation on the residual graph. We repeat this for a constant number of iterations. This completes the description of the algorithm. Its pseudocode is provided in Algorithm \ref{algo2} below. Next, we prove its correctness and analyze running time.

\begin{algorithm}[H]
\caption{Maximum Matching Approximation Algorithm}
\label{algo2}
\begin{algorithmic}[1]
\Procedure{Maximum-Matching-Approximation(Graph $G = (V,E)$)}{}
\State Let $k$ be a positive constant
\State $G' = GeneralBP(G)$
\State $MCMA = \emptyset$
\For {$i = 1,2, ..., k$}
\State $MCMA = MCMA \cup MM(G')$
\State $G = G \setminus MM(G')$ \ \ \ /* remove from $G$ the edges of $MM(G')$, the adjacent edges in $G$ of $MM(G')$, and all isolated vertices. */  
\State $G' = GeneralBP(G)$
\EndFor
\State return MCMA
\EndProcedure
\end{algorithmic}
\end{algorithm}

\begin{theorem}
\label{lemmab}
Given a graph $G=(V,E)$ with bounded neighborhood independence $c$, we achieve a $(2 + \epsilon)$-approximation of the Maximum Matching problem.
\end{theorem}
\begin{proof}
Be Lemma 3.1, after the first iteration of the loop of line 5 of Algorithm 2, we obtain a $(c + 1)$-approximation to maximum matching. Moreover, since each edge in the matching is adjacent in $G'$ to at most $2c$ vertices, and the set of vertices of the matching with their neighbors is $V$, the size of the matching is  at least $|V|/2(c + 1) = n/2(c + 1)$. In each iteration of the loop, at least $1/(c+1)$-fraction of the vertices that are still in $G$ are matched and removed from $G$. Hence, for $i = 1,2,...,$ the number of vertices remaining in $G$ is at most $(c/c + 1)^i \cdot n$. All the other vertices are either matched or have all their neighbors matched. For an arbitrarily small fixed constant $\epsilon$ and a sufficiently large constant $i$, it holds that $(c/c + 1)^i \cdot n \leq \epsilon \cdot n/2(c + 1)$. In other words, the residual set of vertices after $i$ rounds is of size at most an $\epsilon$ fraction of the matching already computed in iteration 1. Thus, after $i$ iteration, any subset of remaining edges of $G$ whose addition makes the result a maximal matching, increases its size by at most an $\epsilon$ fraction. Therefore, the matching after $i$ iterations is a $(1 + \epsilon)$-approximation to MM. Since MM is a 2-approximation to MCM, our algorithm computes a $(2+\epsilon)$-approximate MCM within a constant number of iterations.
\end{proof}

\begin{theorem}
Given a graph $G=(V,E)$ with bounded neighborhood independence $c$, the running time of Algorithm \ref{algo2} is $O(\log^* n)$.
\end{theorem}
\begin{proof}
Using $O(\Delta + \log^* n)$-time Maximal Matching algorithm by Panconesi and Rizzi \cite{panconesi2001some}, and due to the fact that $\Delta(G') = c + 1 = O(1)$, each iteration requires $O(\log^*n)$ rounds. Since the overall number of iterations is constant, the entire running time is $O(\log^* n)$ as well.
\end{proof}

\section{Self-stabilizing Backup Placement in Graphs of Bounded Neighborhood Independence}
\label{selfstab}
In this section we devise a self-stabilizing backup placement algorithm in Dijkstra model of self-stabilization \cite{edsger1974dijkstra}. In this model each vertex has a ROM (Read Only Memory) that is failure free, and a RAM (Random Access Memory) that is prone to failures. An adversary can corrupt the RAM of all processors in any way. However, in certain periods of time, faults do not occur. These periods of time are not known to the processors. The goal of a distributed self-stabilizing algorithm is reaching a proper state in all processors, once faults stop occurring. Since these time points are not known in advance, an algorithm is constantly executed by all processors. The stabilization time is the number of rounds from the beginning of a time period in which faults do not occur, until all processors reach a proper state, given that no additional faults occur during this time period.

Our algorithm stores only the $ID$ of a processor in its ROM. The backup placement selection is stored in the RAM of a processor. The self-stabilizing algorithm is extremely simple. Specifically, In each round each vertex executes the \textit{next-modulo} operation. In other words, each vertex repeats Algorithm \ref{algo1} in each round. This completes the description of the algorithm. Since this operation within a single (faultless) round results in a proper $O(1)$-Backup-placement in a graph with constant neighborhood independence, such an algorithm stabilizes within one round after faults stop occurring. Moreover the solution remains proper as long as there are no faults. We summarize this in the next theorem.
\begin{theorem}
In graphs with neighborhood independence bounded by a constant, our algorithm stabilizes within $1$ round and produces $O(1)$-backup-placement.
\end{theorem}
Thanks to the simplicity of this backup-placement algorithm, it can be used as a building block for other self-stabilization algorithms that employ backup-placements. Specifically, in each round an algorithm can execute the \textit{next-modulo} operation before its own code. This way, starting from the round after the round when faults stop occurring, a subgraph $G'$ of maximum degree $c + 1$ is obtained. This subgraph does not change as long as there are no additional faults. This is because the subgraph is deduced once faults stopped occurring, based only on values in the ROM, and this subgraph does not change in faultless rounds. Thus, a self-stabilizing algorithm with time of the form $f(\Delta,n) = f_1(\Delta) \cdot f_2(n)$ invoked on $G'$ will stabilize within $f(\Delta',n) + 1 = O(f_2(n))$ rounds. This is because $\Delta' = \Delta(G') = c + 1 = O(1)$ in graphs with bounded neighborhood independence $c$. Hence, we obtain the following theorem.
\begin{theorem}
In graph with neighborhood independence at most $c = O(1)$, any self-stabilizing maximal matching algorithm with time $f(\Delta,n) = f_1(\Delta) \cdot f_2(n)$ can be converted to a self-stabilizing $(c + 1)$-approximation of maximum matching with time $O(f_1(c) \cdot  f_2(n)) = O(f_2(n))$.
\end{theorem}
For example, a maximal matching algorithm with running time $O(\Delta n + \Delta^2 \log n)$, such as the self-stabilizing algorithm of \cite{kunne2018self} adapted to a network with IDs, can be converted into a self-stabilizing $(c + 1)$-approximate MCM algorithm with $O(c \cdot n + c^2 \log n)$ time. This is $O( n)$, for $c = O(1)$.

\bibliographystyle{plain} 
\bibliography{bibliography} 

\begin{thebibliography}{10}

\bibitem{assadi2019algorithms}
Sepehr Assadi and Shay Solomon.
\newblock When algorithms for maximal independent set and maximal matching run
  in sublinear time.
\newblock In {\em 46th International Colloquium on Automata, Languages, and
  Programming (ICALP 2019)}. Schloss Dagstuhl-Leibniz-Zentrum fuer Informatik,
  2019.

\bibitem{barenboim2011deterministic}
Leonid Barenboim and Michael Elkin.
\newblock Deterministic distributed vertex coloring in polylogarithmic time.
\newblock {\em Journal of the ACM (JACM)}, 58(5):23, 2011.

\bibitem{barenboim2017deterministic}
Leonid Barenboim, Michael Elkin, and Tzalik Maimon.
\newblock Deterministic distributed (delta+ o (delta))-edge-coloring, and
  vertex-coloring of graphs with bounded diversity.
\newblock In {\em Proceedings of the ACM Symposium on Principles of Distributed
  Computing}, pages 175--184. ACM, 2017.

\bibitem{barenboim2018distributed}
Leonid Barenboim and Tzalik Maimon.
\newblock Distributed symmetry breaking in graphs with bounded diversity.
\newblock In {\em 2018 IEEE International Parallel and Distributed Processing
  Symposium (IPDPS)}, pages 723--732. IEEE, 2018.

\bibitem{barenboim2019fast}
Leonid Barenboim and Gal Oren.
\newblock Fast distributed backup placement in sparse and dense networks.
\newblock {\em arXiv preprint arXiv:1902.08819}, 2019.

\bibitem{edsger1974dijkstra}
W~Edsger.
\newblock Dijkstra.
\newblock {\em Self-stabilizing systems in spite of distributed control.
  Commum. ACM}, 17(11):643--644, 1974.

\bibitem{gfeller2007randomized}
Beat Gfeller and Elias Vicari.
\newblock A randomized distributed algorithm for the maximal independent set
  problem in growth-bounded graphs.
\newblock In {\em Proceedings of the twenty-sixth annual ACM symposium on
  Principles of distributed computing}, pages 53--60. ACM, 2007.

\bibitem{halldorsson2015bp}
Magn{\'u}s~M Halld{\'o}rsson, Sven K{\"o}hler, Boaz Patt-Shamir, and Dror
  Rawitz.
\newblock Distributed backup placement in networks.
\newblock In {\em Proceedings of the 27th ACM symposium on Parallelism in
  Algorithms and Architectures}, pages 274--283. ACM, 2015.

\bibitem{kuhn2019faster}
Fabian Kuhn.
\newblock Faster deterministic distributed coloring through recursive list
  coloring.
\newblock {\em arXiv preprint arXiv:1907.03797}, 2019.

\bibitem{kunne2018self}
Stephan Kunne, Johanne Cohen, and Laurence Pilard.
\newblock Self-stabilization and byzantine tolerance for maximal matching.
\newblock In {\em International Symposium on Stabilizing, Safety, and Security
  of Distributed Systems}, pages 80--95. Springer, 2018.

\bibitem{oren2018distributed}
Gal Oren, Leonid Barenboim, and Harel Levin.
\newblock Distributed fault-tolerant backup-placement in overloaded wireless
  sensor networks.
\newblock In {\em International Conference on Broadband Communications,
  Networks and Systems}, pages 212--224. Springer, 2018.

\bibitem{panconesi2001some}
Alessandro Panconesi and Romeo Rizzi.
\newblock Some simple distributed algorithms for sparse networks.
\newblock {\em Distributed computing}, 14(2):97--100, 2001.

\bibitem{schneider2008log}
Johannes Schneider and Roger Wattenhofer.
\newblock A log-star distributed maximal independent set algorithm for
  growth-bounded graphs.
\newblock In {\em Proceedings of the twenty-seventh ACM symposium on Principles
  of distributed computing}, pages 35--44. ACM, 2008.

\bibitem{schneider2009coloring}
Johannes Schneider and Roger Wattenhofer.
\newblock Coloring unstructured wireless multi-hop networks.
\newblock In {\em Proceedings of the 28th ACM symposium on Principles of
  distributed computing}, pages 210--219. ACM, 2009.

\end{thebibliography}

\end{document}